\documentclass[a4paper]{amsart}
\usepackage[utf8]{inputenc}
\usepackage[a4paper,margin=2.4cm]{geometry}
\usepackage{amsmath,amssymb,amsthm,mathtools,abraces}
\usepackage{color,xcolor}
\usepackage[hypertexnames=false,hyperfootnotes=false,colorlinks=true,linkcolor=blue,%
citecolor=purple,filecolor=magenta,urlcolor=cyan,unicode,linktocpage=true,pagebackref=false]{hyperref}
\usepackage{yhmath}
\usepackage{verbatim}

\newcommand\be{\begin{equation}}
\newcommand\ee{\end{equation}}

\newcommand\p{\partial}

\DeclareMathOperator{\odd}{odd}

\newtheorem{theorem}{Theorem}

\newtheorem*{theorem*}{Theorem}

\usepackage{filecontents}

\title{KdV solves BKP}

\author{Alexander Alexandrov}

\address{IBS Center for Geometry and Physics,
	Pohang University of Science and Technology (POSTECH),
	77 Cheongam-ro, Nam-gu, Pohang, Gyeongbuk, 37673, Korea
}

\email{ {\tt alexandrovsash at gmail.com}}

\begin{document}

\begin{abstract} In this note, we prove that any tau-function of the KdV hierarchy also solves the BKP hierarchy after a simple rescaling of times.
\end{abstract}

\maketitle


\def\thefootnote{\arabic{footnote}}

\setcounter{equation}{0}

In this paper we answer the question about the relation between the KdV and BKP hierarchies, raised in \cite{A}. Namely, in \cite{A} it was observed that several families of the KdV tau-functions important in enumerative geometry and theoretical physics also solve the BKP hierarchy after a simple rescaling of times. Here we prove that {\em any} tau-function of KdV solves the BKP hierarchy.

 In terms of tau-function $\tau_{KP}({\bf t})$ the KP hierarchy, introduced in \cite{SN}, is described by the Hirota bilinear identity
\be\label{HBEKP}
\oint_{\infty} e^{\xi({\bf t-t'},z)}
\tau_{KP} ({\bf t}-[z^{-1}])\tau_{KP} ({\bf t'}+[z^{-1}])dz =0.
\ee
This bilinear identity encodes all nonlinear equations of the KP hierarchy. 
Here we use the standard short-hand notations
\be
{\bf t}\pm [z^{-1}]:= \bigl \{ t_1\pm   
z^{-1}, t_2\pm \frac{1}{2}z^{-2}, 
t_3 \pm \frac{1}{3}z^{-3}, \ldots \bigr \}
\ee
and
\be
\xi({\bf t},z)=\sum_{k>0} t_k z^k.
\ee
If a tau-function of the KP hierarchy does not depend on even time variables,
\be\label{Red}
\frac{\p}{\p t_{2k}}\tau_{KdV}({\bf t})=0 \,\,\,\,\,\,\,\,\,\,\,\, \forall \,k>0,
\ee 
than it is a tau-function of the KdV hierarchy. Since we still have arbitrary $t_{2k}$ and $t_{2k}'$ in the first factor of the integrand in (\ref{HBEKP}),
for the KdV hierarchy the Hirota bilinear identity is given by
\be\label{HBEKP1}
\oint_{\infty} e^{\xi^o({\bf t-t'},z)}\, f(z)\,
\tau_{KdV} ({\bf t}-[z^{-1}])\tau_{KdV} ({\bf t'}+[z^{-1}])dz =0
\ee
for arbitrary series $f(z) \in {\mathbb C}[\![z^2]\!]$. Here we introduce
\be
\xi^o({\bf t},z)= \sum_{k\in{\mathbb Z}_{\odd}^+}t_k z^k.
\ee

In terms of the {\em Hirota derivatives}
\be
P(D_1,D_2,\dots)f\cdot g =\left.P\left(\frac{\p}{\p y_1},\frac{\p}{\p y_2},\dots\right)f({\bf t+y})g({\bf t-y})\right|_{{\bf y}={\bf 0}},
\ee
the first few equations of the KdV hierarchy can be represented as
\begin{equation}
\begin{split}\label{HirKdV}
\left(D_1^4-4D_1D_3\right)\tau_{KdV}\cdot \tau_{KdV}=0,\\
\left(D_1^6+4D_1^3D_3-32D_3^2\right)\tau_{KdV}\cdot \tau_{KdV}=0,\\
\left(D_1^6-20D_1^3D_3-80D_3^2+144D_1D_5\right)\tau_{KdV}\cdot \tau_{KdV}=0.\\
\end{split}
\end{equation}

The BKP hierarchy was introduced by Date, Jimbo, Kashiwara and Miwa in \cite{JMBKP}. 
It can be represented in terms of tau-function $\tau_{BKP}({\bf t})$ by the Hirota bilinear identity similar to (\ref{HBEKP}):
\be\label{HBE}
\frac{1}{2 \pi i}\oint_{\infty} e^{\xi^o({\bf t}-{\bf t}',z)}
\tau_{BKP} ({\bf t}-2[z^{-1}])\tau_{BKP} ({\bf t'}+2[z^{-1}])\frac{dz}{z} =\tau_{BKP} ({\bf t})\tau_{BKP} ({\bf t'}).
\ee
The first equation of the BKP hierarchy in terms of the Hirota derivatives is given by
\be
\left(D_1^6-5D_1^3D_3-5D_3^2+9D_1D_5\right)\tau_{BKP}\cdot \tau_{BKP}=0.
\ee

Note, that if we map $D_k \mapsto D_k/2$ then this equation coincides with the last equation in (\ref{HirKdV}). Let us show that this is more than just a coincidence, and all equations of the BKP hierarchy follow from the KdV hierarchy. The main result of this note is 
\begin{theorem}\label{MT}
For any KdV tau-function  
\be\label{ms}
\tau_{BKP}({\bf t})=\tau_{KdV}({\bf t}/2)
\ee 
is a tau-function of the BKP hierarchy.
\end{theorem}
\begin{proof}
Let us consider 
\be
f(z)=\frac{e^{\xi^o({\bf t}-{\bf t}',z)}-e^{-\xi^o({\bf t}-{\bf t}',z)}}{z}  \in {\mathbb C}[\![z^2]\!],
\ee
then the Hirota bilinear identity (\ref{HBEKP1}) for the KdV hierarchy leads to 
\be
\frac{1}{2 \pi i}\oint_{\infty} e^{2\xi^0({\bf t-t'},z)} 
\tau_{KdV} ({\bf t}-[z^{-1}])\tau_{KdV} ({\bf t'}+[z^{-1}])\frac{dz}{z} =\tau_{KdV}  ({\bf t})\tau_{KdV}  ({\bf t'}).
\ee
Now we substitute the times ${\bf t} \mapsto {\bf t}/2$, ${\bf t}' \mapsto {\bf t}'/2$:
\be
\frac{1}{2 \pi i}\oint_{\infty} e^{\xi^o({\bf t}-{\bf t}',z)}
\tau_{KdV} ({\bf t}/2-[z^{-1}])\tau_{KdV} ({\bf t'}/2+[z^{-1}])\frac{dz}{z} =\tau_{KdV}  ({\bf t}/2)\tau_{KdV}  ({\bf t'}/2).
\ee
This equation coincides with the Hirota bilinear identity for the BKP hierarchy (\ref{HBE}) if we identify tau-functions by (\ref{ms}). This completes the proof.
\end{proof}
It is easy to see that the converse statement is false, that is, the KdV hierarchy is a certain reduction of the BKP hierarchy. This reduction, as well as similar description of the higher Gelfand--Dickey hierarchies, will be discussed elsewhere. 

This theorem partially explains the results of \cite{A,MM}. Indeed, in \cite{MM} it was noted that the Kontsevich-Witten tau-function - one of the most important solutions of the KdV hierarchy - has a simple expansion in terms of the Schur Q-functions. This result was generalized in
\cite{A} to a family of KdV tau-functions related to the Br\'ezin--Gross--Witten model. On the basis of these expansions the author has conjectured (for the Kontsevich-Witten tau-function) and proved (for the Br\'ezin--Gross--Witten tau-function) that these KdV tau-functions also solve the BKP hierarchy. From Theorem \ref{MT} it follows, that this part of the conjectures in \cite{A} trivially follows from a general relation between the KdV and BKP hierarchies. Moreover, this theorem makes the Schur Q-functions a natural basis for expansion of the KdV tau-functions.

There are several different ways to relate KdV and BKP hierarchies. In particular, in \cite{JMred} the authors describe an identification of KdV hierarchy with the 4-reduction of BKP, see Remark on page 1098. Another reduction from BKP to KdV (so-called 1-constrained BKP) was described in \cite{Orl1,Orl,Cheng}. To the best of our understanding, these reductions do not coincide with the one considered in this paper.
 
 \section*{Acknowledgments}
The author is grateful to J. van de Leur, A. Mironov, A. Morozov, and A. Orlov for useful discussions.  This work was supported by IBS-R003-D1.

\bibliographystyle{alpha}
\bibliography{KPTRref}

\end{document}